\documentclass[final,12pt]{article}
\usepackage[utf8]{inputenc}
\usepackage[T1]{fontenc}
\usepackage[colorlinks=true,linkcolor=blue,citecolor=magenta]{hyperref}
\usepackage{amssymb,amsmath,amsthm}
\usepackage{cite}
\usepackage{color}
\usepackage{enumerate}
\usepackage{physics}
\usepackage{lmodern}
\usepackage{microtype}
\usepackage[conditional,light,first,bottomafter]{draftcopy}
\draftcopyName{DRAFT\space\today}{130}
\draftcopySetScale{65}
\usepackage[letterpaper,hmargin=3.1cm,vmargin=3.2cm]{geometry}
\usepackage{setspace}
\makeatletter
\renewcommand{\section}{\@startsection%
{section}%
{1}%
{0em}%
{1.7em}%
{1.2em}%
{\normalfont\large\centering\bfseries}}
\renewcommand{\@seccntformat}[1]%
{\csname the#1\endcsname.\hspace{0.5em}}
\makeatother

\numberwithin{equation}{section}
\newtheorem{theorem}{Theorem}[section]

\newtheorem{lemma}{Lemma}[section]
\newtheorem{corollary}{Corollary}[section]
\theoremstyle{definition}
\newtheorem{definition}{Definition}
\newtheorem{remark}{Remark}
\newtheorem*{notation}{Notation}

\newcommand{\inner}[2]{\left\langle#1,#2\right\rangle}
\newcommand{\cc}[1]{\overline{#1}}
\newcommand{\reals}{\mathbb{R}}
\newcommand{\nats}{\mathbb{N}}
\newcommand{\complex}{\mathbb{C}}
\newcommand{\convergesto}[1]{\xrightarrow[#1\to\infty]{}}

\DeclareMathOperator{\dom}{dom}
\DeclareMathOperator{\ran}{ran}
\DeclareMathOperator{\dist}{dist}

\DeclareMathOperator{\diag}{diag}
\begin{document}
\begin{titlepage}
\title{Green matrix estimates of block Jacobi matrices I:
  Unbounded gap in the essential spectrum
\footnotetext{%
Mathematics Subject Classification(2010):
39A22,  
47B36, 
33E30  
}
\footnotetext{%
Keywords:
Block Jacobi operators;
Generalized eigenvectors;
Decay bounds.
}
}
\author{
\textbf{Jan Janas}
\\
\small Institute of Mathematics\\[-1.6mm]
\small Polish Academy of Sciences (PAN)\\[-1.6mm]
\small Ul. Sw. Tomasza 30, 31-027\\[-1.6mm]
\small Krakow, Poland\\[-1.6mm]
\small \texttt{najanas@cyf-kr.edu.pl}
\\[2mm]
\textbf{Sergey Naboko}
\\
\small Department of Mathematical Physics\\[-1.6mm]
\small Institute of Physics\\[-1.6mm]
\small St. Petersburg State University\\[-1.6mm]
\small Ulyanovskaya 1, St. Petersburg 198904, Russia\\[-1.6mm]
\small \texttt{sergey.naboko@gmail.com}
\\[2mm]
\textbf{Luis O. Silva}
\\
\small Departamento de F\'{i}sica Matem\'{a}tica\\[-1.6mm]
\small Instituto de Investigaciones en Matem\'aticas Aplicadas y en Sistemas\\[-1.6mm]
\small Universidad Nacional Aut\'onoma de M\'exico\\[-1.6mm]
\small C.P. 04510, Ciudad de M\'exico\\[-1.6mm]
\small \texttt{silva@iimas.unam.mx}
}
\date{}
\maketitle
\vspace{-4mm}
\begin{center}
\begin{minipage}{5in}
  \centerline{{\bf Abstract}} \bigskip This work deals with decay
  bounds for Green matrices and generalized eigenvectors of block
  Jacobi matrices when the real part of the spectral parameter lies in
  an infinite gap of the operator's essential spectrum. We consider
  the cases of commutative and noncommutative matrix entries
  separately. An example of a block Jacobi operator with
  noncommutative entries and nonnegative essential spectrum is given
  to illustrate the results.
\end{minipage}
\end{center}
\thispagestyle{empty}
\end{titlepage}
\section{Introduction}
\label{sec:intro}

In this work, we consider block Jacobi operators acting in
$\mathcal{H}=l_2(\nats,\mathcal{K})$, the space of
square-summable sequences whose elements lie in a Hilbert space
$\mathcal{K}$ (see the precise definition in the next section). As in
the case of scalar Jacobi operators, block Jacobi operators are
associated with a second order difference equation, but instead of
having scalar coefficients, this equation has operator
coefficients (see~\eqref{eq:difference-expr}). These operators, $A_n$
and $B_n=B_n^*$ ($n\in\nats$), are the entries of a block Jacobi
matrix (see \eqref{eq:block-jm}). The class of block Jacobi operators
under consideration, which is generically denoted by $J$, is 
 such that the
operators $A_n$ and $B_n$ are bounded and defined on the whole Hilbert
space $\mathcal{K}$, that is, $A_n,B_n\in B(\mathcal{K})$ for all
$n\in\nats$. Additionally, we require that $J$ is self-adjoint and
semi-bounded with its essential spectrum
lying inside the interval $[b,+\infty)$. Under these assumptions, we
provide estimates for the decay of the matrix entries of $(J-\lambda
I)^{-1}$, i.\,e. the entries of the Green matrix of $J$, when
$\Re\lambda<b$. This in particular gives estimates of the so-called
generalized eigenvectors. Moreover, when
$\sigma(J)\cap(-\infty,b)\ne\emptyset$, we show that these estimates
apply for the eigenvectors of $J$ corresponding to eigenvalues below
$b$.

Similar questions for estimates of the generalized eigenvectors of
scalar Jacobi matrices have been addressed in previous papers
\cite{MR2480099,MR3028179} (references for results on this matter,
preceding the ones of \cite{MR2480099,MR3028179}, can be found in
these papers). In \cite{MR2480099}, the case where the scalar Jacobi
matrix $J$ satisfies the operator inequality $J\ge bI$ is studied. In
the other work \cite{MR3028179}, the presence of a bounded discrete
spectrum of $J$ is allowed. In
  \cite{MR2480099,MR3028179}, there is a refinement of the
  Combes-Thomas method for obtaining estimates of the Green function
  which provides sharp coefficient estimates and establishes that the
  bounds depend inversely on the growth of the off-diagonal
  entries. In this work, on the basis of the methods developed in
  \cite{MR2480099,MR3028179}, we establish estimates that include as a
  particular case the bounds found earlier in
  \cite{MR2480099,MR3028179}. The growth of the off-diagonal
  operators refers either to the growth of the operator norm (see
  Theorems~\ref{thm:infinite-interval} and
  \ref{thm:infinite-interval-in-spectrum}) or to the absolute value of
  the operator (see
  Theorem~\ref{thm:infinite-interval-commutation}). The latter case
  permits a wide range of interesting examples.

It should be mentioned that the
semi-boundedness from below of $J$ is crucial as it already was in
the scalar case. Indeed, in \cite{MR2579689}, an example of a scalar
Jacobi matrix was produced with essential spectrum covering the
interval $(0,\infty)$ and negative spectrum being discrete and
accumulating to $-\infty$. For this example, the explicitly calculated
asymptotics of generalized eigenvectors for $\lambda<0$ does not
satisfy the estimates of
Theorem~\ref{thm:infinite-interval-in-spectrum}. Namely, the
asymptotics found in \cite{MR2579689} contains non-removable
information on the main diagonal of the scalar Jacobi matrix while the
estimates given by Theorem~\ref{thm:infinite-interval-in-spectrum} do
not depend on the main diagonal of the block Jacobi matrix.
The fact that the asymptotic behavior of generalized eigenvectors of a
semi-bounded Jacobi operator is independent of its main diagonal was
known in the scalar case
\cite{MR3028179}. Theorems~\ref{thm:infinite-interval} and
\ref{thm:infinite-interval-in-spectrum} extend that result to block
Jacobi operators.

The growth estimates of generalized eigenvectors found in this work
provide a generalization of Combes-Thomas type estimates
\cite{MR1441595,MR0391792} applicable to various discrete random
models, including the Anderson model, with matrix potentials (an
accessible and detailed survey on the matter in the scalar case is
found in \cite{MR2509110} and for Combes-Thomas estimates see Chapter
11 there). It is pertinent to mention here that there are various works dealing
with random block-type operators, for instance
\cite{MR3210957,MR3077279,MR969209}, and some of them treating the
problem of 
localization \cite{MR3210957,MR3077279}. In the context of our work,
\cite[Lem.\,5.7]{MR3077279} is of particular relevance since it gives
a Combes-Thomas type estimate for random block operators.

Apart from random block-type operators, the estimates of generalized
eigenvectors in the block Jacobi matrix case may have other
interesting applications. One of them is related to the investigation
of the spectral phase transition phenomena of the second kind (see
\cite{MR2480099} for a case of one-threshold transition) for discrete
models with matrix entries. Block Jacobi matrices permit more freedom
to construct models which exhibit multi-threshold spectral phase
transitions.

  Spectral properties of block Jacobi operators have
  also been studied in \cite{MR2591198} which carries out averaging of the
  spectral measure over boundary conditions. More recently,
  \cite{MR3767432} developed a framework which simplifies the general
  local Green function relations found in \cite{MR2184188} and
  illustrates the power of the transfer matrix method.

In this work we deal exclusively with the case of
  unbounded gap in the essential spectrum.
The case of a bounded gap in the essential spectrum of block Jacobi
operators requires a different technique and will be considered in a
forthcoming paper.

The following is a summary of the
paper. Section~\ref{sec:block-jacobi} contains a short survey on the
basics of the theory of block Jacobi
matrices. Section~\ref{sec:estim-gener-eigenv} presents the main result
of the paper (Theorem~\ref{thm:infinite-interval}) which gives the
estimates of the Green matrix entries
decay. Theorem~\ref{thm:infinite-interval-in-spectrum} also shows
estimates of the eigenvector's decay for the eigenvalue
$\lambda<b$. Note that for such $\lambda$, the Green matrix, i.\,e.
the resolvent of $J$, is not defined, but this problem can be easily
overcome by a proper ``small change'' of the matrix $J$.
Section~\ref{sec:case-comm-entr}
deals with the special case of commuting entries. In this case, more
detailed estimates, counting the matrix character of the entries, are
obtained (Theorem~\ref{thm:infinite-interval-commutation} and
Corollary \ref{cor:norm-to-infty-commutation}). This contrasts the
results of Section~\ref{sec:estim-gener-eigenv}, where just the norm
of the entries are involved in the
estimates. Section~\ref{sec:an-example-noncommuting} presents an
example of a block Jacobi matrix with a $2\times 2$ matrix entries
with nonnegative essential spectrum. The application of
Theorem~\ref{thm:infinite-interval-in-spectrum} and a 
heuristic analysis of solutions (the Levinson type asymptotics form)
give a sort of arguments for proving the sharpness of
Theorem~\ref{thm:infinite-interval-in-spectrum}. We hope that the true
block matrix example of Section 5 is also of independent interest. It
exhibits the matrix character role of the entries.





\section{Block Jacobi matrices}
\label{sec:block-jacobi}
\begin{notation}
  \label{not:basic-notation}
  The following notation is used throughout this work.
\begin{enumerate}[(I)]
\item
 \label{not-1}
  By $\mathcal{H}$, we denote a separable infinite dimensional
  Hilbert space. This space admits the decomposition
  \begin{equation*}
    \mathcal{H}=\bigoplus_{m=1}^\infty\mathcal{K}_m\,,
  \end{equation*}
  where, for all $m\in\nats$, $\mathcal{K}_m=\mathcal{K}$ and
  $\mathcal{K}$ is either an infinite or finite dimensional subspace
  of $\mathcal{H}$. Therefore, $\mathcal{K}$ is either unitarily
  equivalent to $l_2(\nats)$ or $\complex^d$ with a fixed
  $d\in\nats$. Here and throughout the text, $l_2(\nats)$ stands for
  the space of infinite, square-summable complex sequences. The
  Hilbert space $\mathcal{H}$ is usually denoted by $l_2(\nats,\mathcal{K})$.

\item\label{not-2} The symbol $\norm{\cdot}$ is used to denote the norm in
  $\mathcal{H}$, while the norm in $\mathcal{K}$ is denoted by
  $\norm{\cdot}_{\mathcal{K}}$. $B(\mathcal{H}$) and $B(\mathcal{K})$
  denote the spaces of bounded linear operators defined on the whole
  space $\mathcal{H}$ and $\mathcal{K}$, respectively. The norms in
  $B(\mathcal{H})$ and $B(\mathcal{K})$ are denoted by
  $\norm{\cdot}_{B(\mathcal{H})}$ and $\norm{\cdot}_{B(\mathcal{K})}$,
  respectively.

\item \label{not-3}
A vector $u$ in
$\mathcal{H}$ can be written as a sequence
\begin{equation}
  \label{eq:sequence}
  u=\{u_m\}_{m=1}^\infty\,,\qquad
u_m\in\mathcal{K}_m\,,
\end{equation}
where $\sum_{m=1}^\infty\norm{u_m}_{\mathcal{K}}^2<+\infty$.
We also use the notation
\begin{equation*}
  u=(u_1,u_2,u_3,\dots)^{\rm\textsf{T}}
\end{equation*}

\item \label{not-4} Throughout this work, we use $I$ to denote the
  identity operator in the spaces $\mathcal{H}$ and $\mathcal{K}$
  since it will cause no confusion to use the same letter for these
  operators.  The orthogonal projector in $\mathcal{H}$ onto the
  subspace $\mathcal{K}_m$ is denoted by $P_m$ while the symbol
  $\widetilde{P}_M$ stands for the orthogonal projector onto
  $\bigoplus_{m=1}^M\mathcal{K}_m$.
\item \label{not-5}
Given a closed, densely defined operator $A$ in a Hilbert space, we
denote by $\abs{A}$ the operator $(A^*A)^{1/2}$.
\end{enumerate}
\end{notation}

Let us turn to the definition of block Jacobi operators. For any
sequence (\ref{eq:sequence}), consider the second order difference
expressions
\begin{subequations}
  \label{eq:difference-expr}
\begin{align}
   \label{eq:difference-recurrence}
  (\Upsilon u)_k&:= A_{k-1}^* u_{k-1} + B_k u_k + A_ku_{k+1}
  \quad k \in \mathbb{N} \setminus \{1\},\\
   \label{eq:difference-initial}
   (\Upsilon u)_1&:=B_1 u_1 + A_1 u_2\,,
\end{align}
\end{subequations}
where $B_k=B_k^*, A_k\in B(\mathcal{K})$ for any $k\in\mathbb{N}$.
\begin{definition}
 \label{def:j-nought}
  In $\mathcal{H}$, define the operator $J_0$ whose
  domain is the set of sequences (\ref{eq:sequence}) having a finite number of
  non-zero elements and is given by $J_0f:=\Upsilon
  f$. Since
$J_0$ is symmetric (therefore closable), one can consider its closure
which is denoted by $J$.
\end{definition}

We have defined the operator $J$ so that the block tridiagonal matrix
\begin{equation}
  \label{eq:block-jm}
\begin{pmatrix}
B_1&A_1&0&0&\cdots\\
A_1^*&B_2&A_2&0&\\
0&A_2^*&B_3&A_3&\ddots\\
0&0&A_3^*&B_4&\ddots\\
\vdots&&\ddots&\ddots&\ddots
\end{pmatrix}
\end{equation}
can be regarded as the matrix representation of the operator $J$
(see \cite[Sec. 47]{MR1255973} where a definition of the matrix
representation of an unbounded symmetric operator is given and
also \cite{MR2379691} and \cite{arXiv:1705.06138} for general
questions on block Jacobi operators).

In this work, we impose conditions on the sequences
$\{A_m\}_{m=1}^\infty$ and $\{B_m\}_{m=1}^\infty$ so that $J$ is
self-adjoint. A sufficient condition for this to happen is the
generalized Carleman criterion \cite[Chap.7 Thm. 2.9]{MR0222718},
viz., if $\sum_{m=1}^\infty 1/\norm{A_m}_{\mathcal{K}}=+\infty$, then
$J$ is self-adjoint.

Self-adjointness of $J$ implies that its domain, $\dom(J)$, coincides
with the maximal linear set in which the result of the ``action'' of the matrix
(\ref{eq:block-jm}) on a sequence (\ref{eq:sequence}) yields a sequence
in $\mathcal{H}$. Thus, it cannot lead to confusion if we use the same
letter $J$ to denote both the operator and the matrix. Likewise,
\begin{equation*}
  \diag\{C_m\}_{m=1}^\infty:=
\begin{pmatrix}
C_1&0&0&0&\cdots\\
0&C_2&0&0&\\
0&0&C_3&0&\ddots\\
0&0&0&C_4&\ddots\\
\vdots&&\ddots&\ddots&\ddots
\end{pmatrix}\,,
\end{equation*}
where $C_m\in B(\mathcal{K})$ for any $m\in\nats$,
is used for denoting the operator and the matrix (the operator being
$\bigoplus_{m=1}^\infty C_m$ with $C_m\in B(\mathcal{K})$ for all $m\in\nats$).

Consider the unilateral vector shift operator $S$ in $\mathcal{H}$ given by
\begin{equation*}
  S(u_1,u_2,u_3,\dots)^{\rm\textsf{T}}=(0,u_1,u_2,\dots)^{\rm\textsf{T}} 
\end{equation*}
and its adjoint $S^*$ for which
\begin{equation*}
  S^*(u_1,u_2,u_3,\dots)^{\rm\textsf{T}}=(u_2,u_3,u_4,\dots)^{\rm\textsf{T}}
\end{equation*}
It can be verified that the operator
\begin{equation}
  \label{eq:representation-j}
  \diag\{B_m\}_{m=1}^\infty + S\diag\{A^*_m\}_{m=1}^\infty +
\diag\{A_m\}_{m=1}^\infty S^*
\end{equation}
coincides with the self-adjoint operator $J$.


\begin{definition}
  \label{def:green-function}
  Assume that the operator $J$ given in Definition~\ref{def:j-nought}
  is self-adjoint. For any $\lambda$ in the resolvent set of $J$, define
  \begin{equation*}
    G_{jk}(\lambda):=P_j(J-\lambda I)^{-1}P_k\,.
  \end{equation*}
\end{definition}
Note that by this definition $G_{jk}(\lambda)^*=G_{kj}(\cc{\lambda})$,
and therefore
\begin{equation}
  \label{eq:norm-resolvent-adjoint}
  \norm{G_{jk}(\lambda)}_{B(\mathcal{K})}=
\norm{G_{kj}(\cc{\lambda})}_{B(\mathcal{K})}\,.
\end{equation}
Due to the fact that for $\lambda\not\in\sigma(J)$,
$(J-\lambda I)^{-1}\in B(\mathcal{H})$, one verifies that
\begin{equation*}
  (J-\lambda I)^{-1}u=
\left(\sum_{k=1}^\infty G_{1k}u_k, \sum_{k=1}^\infty G_{2k}u_k,
\sum_{k=1}^\infty G_{3k}u_k, \dots\right)^{\rm\textsf{T}}
\end{equation*}
Thus, $G_{jk}(\lambda)$ can be regarded as the entries of the matrix
representation of the resolvent of $J$ at $\lambda$. We refer to
$\{G_{jk}(\lambda)\}_{j,k=1}^\infty$ as the block Green matrix
corresponding to $J$ at $\lambda$.

\section{Estimates of generalized eigenvectors in an
  unbounded gap of the essential spectrum}
\label{sec:estim-gener-eigenv}
In this section we find estimates for generalized eigenvectors of the
operator $J$ given by Definition~\ref{def:j-nought} when it is
semi-bounded. There is no loss of generality in assuming the operator
$J$ bounded from below. We consider that the real part of the spectral
parameter is below the essential spectrum and obtain estimates for
both cases: when this parameter is not an eigenvalue and when it is.

To simplify the writing of some formulae, we introduce the functions
\begin{equation}
  \label{eq:psi-function-def}
   \psi(x):=x^2 e^{x}\,,\quad 0\le x\,.
\end{equation}
and
\begin{equation}
  \label{eq:phi-function-def}
    \phi_\delta(x):=
   \begin{cases}
     1/\sqrt{\delta} & \text{ if } 0\le x<\delta\\
     1/\sqrt{x}& \text{ if } \delta\le x
   \end{cases}
\end{equation}
\begin{theorem}
  \label{thm:infinite-interval}
  Assume that the operator $J$ given in Definition~\ref{def:j-nought}
  is self-adjoint and bounded from below. Take a real
  number $b$ such that
  $(-\infty,b)\cap\sigma_{ess}(J)=\emptyset$ and consider a complex
  number $\lambda$ with $\Re\lambda<b$. Fix $\delta>0$ and
  $\epsilon$ arbitrarily small in $(0,1)$.
 If $\lambda\not\in\sigma(J)$, then
    \begin{equation*}
      \norm{G_{jk}(\lambda)}_{B(\mathcal{K})}\le
    C
   \exp(-\gamma(\lambda)\!\!\sum\limits_{m=\min(j,k)}^{\max(j,k)-1}
\phi_\delta(\norm{A_m}_{B(\mathcal{K})}))\,,
\end{equation*}
where $G_{jk}(\lambda)$ is given in
Definition~\ref{def:green-function}, $C$ does not depend on $j$ and
$k$, $\phi_\delta$ is given in \eqref{eq:phi-function-def}
and
\begin{equation}
\label{eq:def-gamma}
  \gamma(\lambda):=
\sqrt{\delta}\psi^{-1}\left(\frac{(b-\Re\lambda)(1-\epsilon)}{\delta}\right)
\end{equation}
with $\psi^{-1}$ being the inverse function of $\psi$ given in
\eqref{eq:psi-function-def}.
\end{theorem}
\begin{proof}
Let $E$ be the spectral measure of the self-adjoint operator $J$,
i.\,e., $J=\int sdE_s$. Define
\begin{equation*}
  K:=(J-b)E(-\infty,b)\,.
\end{equation*}
Due to the fact that there is no essential spectrum in $(-\infty,b)$,
the operator $K$ is compact and
\begin{equation}
  \label{eq:j-d-greater-d}
  J-K\ge bI\,.
\end{equation}
Choose the number $M$ so large that
\begin{equation}
  \label{eq:estimate-k-projected}
  \norm{K(I-\widetilde{P}_M)}_{B(\mathcal{H})}\le (b-\Re\lambda)\frac{\epsilon}{2}
\end{equation}
for any fixed $\lambda$ such that $\Re\lambda<b$. This can be done
since the sequence $\{\widetilde{P}_M\}_{M=1}^\infty$ converges strongly to $I$ as
$M\to\infty$. Clearly, our choice of $M$ depends on
$\lambda,b,\epsilon$, and $J$.

For any fixed $N\in\nats$, let
\begin{equation}
\label{eq:def-Phi-m}
  \Phi_m:=
\begin{cases}
\exp\left(-\gamma\sum_{k=1}^{m-1}\phi_\delta(\norm{A_k}_{B(\mathcal{K})})\right)I\,,&
m\le N\,,\\
\exp\left(-\gamma\sum_{k=1}^{N-1}\phi_\delta(\norm{A_k}_{B(\mathcal{K})})\right)I\,,&
m> N\,,
\end{cases}
\end{equation}
where $\phi_\delta$ is given in \eqref{eq:phi-function-def} and $\gamma$ is to be
determined later. Note that $\Phi_m$ is a scalar matrix for all
$m\in\nats$.  Consider the following bounded operator in $\mathcal{H}$
\begin{equation*}
\Phi:=\diag\{\Phi_m\}_{m=1}^\infty\,.
\end{equation*}
Besides depending on $\delta$ and $\gamma$, this operator depends on
$N$. When needed, we indicate this dependence explicitly, i.\,e.,
$\Phi=\Phi(N)$. Note that, by freezing the sequence
$\{\Phi_m\}_{m=1}^\infty$ from $\Phi_N$ onwards in
\eqref{eq:def-Phi-m}, the operator $\Phi(N)$ is a boundedly invertible
contraction for any finite $N$. At the end of this proof, we let $N\to
+\infty$.

Define
\begin{equation*}
  F:=S\diag\{\Phi_{m+1}^{-1}A_m^*\Phi_{m}-A_m^*\}
+\diag\{\Phi_m^{-1}A_m\Phi_{m+1}-A_m\}S^*\,.
\end{equation*}
By \eqref{eq:def-Phi-m}, $F\in B(\mathcal{H})$ (actually, it is a
``block-finite-rank'' operator, viz., the sequences of the form
\eqref{eq:sequence} in the range of $F$ have a finite number of
nonzero elements). Using (\ref{eq:representation-j}), one verifies that
\begin{equation}
\label{eq:relation-f-phi}
  F=\Phi^{-1}J\Phi - J
\end{equation}
and
\begin{equation}
 \label{eq:phi-operator-phi}
  \Phi^{-1}(J-\lambda I)\Phi=J+F-\lambda I\,.
\end{equation}
Also,
\begin{align*}
  2\Re F&=F+F^*\\ &=
S\diag\{\Phi_m^{-1}A_m\Phi_{m+1}-2A_m+\Phi_m^*A_m(\Phi_{m+1}^*)^{-1}\}^*\\
 &+\diag\{\Phi_m^{-1}A_m\Phi_{m+1}-2A_m+\Phi_m^*A_m(\Phi_{m+1}^*)^{-1}\}S^*\,.
\end{align*}
Since the matrix of the operator $2\Re F$ has only two block diagonals
not necessarily zero and one diagonal is the adjoint of the other, one
has
\begin{equation}
\label{eq:real-part-first-estimate}
\norm{\Re F}_{B(\mathcal{H})}\le\sup_{m\in\nats}\left\{\norm{\Phi_m^{-1}A_m\Phi_{m+1}-2A_m+
\Phi_m^*A_m(\Phi_{m+1}^*)^{-1}}_{B(\mathcal{K})}\right\}\,.
\end{equation}
Let us show that, by choosing $\gamma$ appropriately, one can ensure that
\begin{equation}
 \label{eq:realpart-c}
  \norm{\Re F}_{B(\mathcal{H})}\le (1-\epsilon)(b-\Re\lambda)
\end{equation}
under our assumption that $\Re\lambda<b$.
First note that
(\ref{eq:real-part-first-estimate}) implies
\begin{align*}
  \norm{\Re F}_{B(\mathcal{H})}&\le
\sup_{m\le N}
\left\{\norm{A_m\left(e^{-\gamma\phi_{\delta}(\norm{A_m}_{B(\mathcal{K})})}
-2I+e^{\gamma\phi_\delta(\norm{A_m}_{B(\mathcal{K})})}\right)}_{B(\mathcal{K})}\right\}
\\
&\le
\sup_{m\in\nats}
\left\{\norm{A_m\left(e^{-\gamma\phi_{\delta}(\norm{A_m}_{B(\mathcal{K})})}
-2I+e^{\gamma\phi_\delta(\norm{A_m}_{B(\mathcal{K})})}\right)}_{B(\mathcal{K})}\right\}\,.
\end{align*}
On the basis of the inequality
\begin{equation}
 \label{eq:algebraic-inequality-exp}
  0\le e^x-2+e^{-x}\le x^2e^x
\end{equation}
valid for $x\ge 0$, one has
\begin{equation}
 \label{eq:bound-real-part-through-sup}
  \norm{\Re F}_{B(\mathcal{H})}\le
\sup_{m\in\nats}
\left\{\norm{A_m}_{B(\mathcal{K})}\gamma^2\phi_\delta^2(\norm{A_m}_{B(\mathcal{K})})
e^{\gamma\phi_\delta(\norm{A_m}_{B(\mathcal{K})})}\right\}\,.
\end{equation}
Fix $\delta>0$ and choose $\gamma$ so small that the inequality
\begin{equation}
  \label{eq:aux-inequality}
  \xi\psi(\gamma\phi_\delta(\xi))\le (1-\epsilon)(b-\Re\lambda)
\end{equation}
holds for all $\xi$ in $[0,\delta]$. Taking into account the behaviour of the
function $\phi_\delta$ when its argument is not greater than $\delta$
(see \eqref{eq:phi-function-def}),
one verifies that the inequality holds whenever $\gamma$ is given by
(\ref{eq:def-gamma}). Actually the inequality also holds for
$\delta<\xi$ since the function $\xi\psi(\gamma\phi_\delta(\xi))$ is
monotone decreasing in $\xi$ for $\xi<\delta$ as long as $\gamma>0$.
In view of \eqref{eq:psi-function-def} and \eqref{eq:phi-function-def}, the inequalities
(\ref{eq:bound-real-part-through-sup}) and (\ref{eq:aux-inequality})
imply (\ref{eq:realpart-c}).

Using (\ref{eq:j-d-greater-d}), one verifies that
\begin{equation*}
  \Re(J+F-K\widetilde{P}_M-\lambda I)\ge
(b-\Re\lambda)I+K(I-\widetilde{P}_M +\Re F)\,.
\end{equation*}
Thus, due to \eqref{eq:estimate-k-projected} and
\eqref{eq:realpart-c}, the last inequality  yields
\begin{equation*}
  \Re(J+F-K\widetilde{P}_M-\lambda I)\ge
  (b-\Re\lambda)\left(1-\frac{\epsilon}{2}-(1-\epsilon)\right)I =
  (b-\Re\lambda))\frac{\epsilon}{2} I\,.
\end{equation*}
Since $J+F-K\widetilde{P}_M$ is a self-adjoint operator perturbed by
a bounded operator, it follows from the last inequality that
\begin{equation*}
  J+F-K\widetilde{P}_M-(\lambda + \frac{\epsilon}{2}(b -\Re\lambda))I
\end{equation*}
is a maximal accretive operator. Thus, $J+F-K\widetilde{P}_M-\lambda I$ is
invertible and the estimate
(see
\cite{MR0407617})
\begin{equation}
  \label{eq:estimate-resolvent-aux}
  \norm{(J+F-K\widetilde{P}_M-\lambda I)^{-1}}_{B(\mathcal{H})}\le \frac{2}{\epsilon(b-\Re\lambda)}
\end{equation}
is known to hold for $\Re\lambda< b$.

Below we use that the operator $I+K\widetilde{P}_M((J+F-K\widetilde{P}_M-\lambda I)^{-1}$ is
boundedly invertible. This fact is established as follows. Since
$K\widetilde{P}_M((J+F-K\widetilde{P}_M-\lambda I)^{-1}$ is compact, it suffices to show that
$\ker(I+K\widetilde{P}_M((J+F-K\widetilde{P}_M-\lambda I)^{-1})$ is trivial. Suppose on the
contrary that
\begin{equation*}
  0\ne v\in\ker(I+K\widetilde{P}_M(J+F-K\widetilde{P}_M-\lambda I)^{-1})\,,
\end{equation*}
then
\begin{equation*}
  \left(J+F-\lambda I\right)(J+F-K\widetilde{P}_M-\lambda I)^{-1}v=0\,,
\end{equation*}
which implies that $\ker(J+F-\lambda I)$ is not empty since
$(J+F-K\widetilde{P}_M-\lambda I)^{-1}v\ne 0$. Therefore, using
\eqref{eq:phi-operator-phi} and taking into account that $\Phi(N)$ and
$[\Phi(N)]^{-1}$ are bounded for any $N<+\infty$, one concludes that
$J-\lambda I$ is not invertible which contradicts the fact that
$\lambda$ is in the resolvent set.

Due to the algebraic identity
\begin{equation}
  \label{eq:aux-inversion}
  \left(J+F-\lambda I\right)^{-1}=
(J+F-K\widetilde{P}_M-\lambda I)^{-1}\left[I+K\widetilde{P}_M(J+F-K\widetilde{P}_M-\lambda I)^{-1}\right]^{-1}
\end{equation}
and \eqref{eq:phi-operator-phi}, one has
\begin{align*}
  &K\widetilde{P}_M\Phi^{-1}(J-\lambda I)^{-1}\Phi\\
&=K\widetilde{P}_M(J+F-K\widetilde{P}_M-\lambda
  I)^{-1}[I+K\widetilde{P}_M(J+F-K\widetilde{P}_M-\lambda I)^{-1}]^{-1}\\
&=(\!-I\!+\!I+\!K\widetilde{P}_M(J+F-K\widetilde{P}_M-\lambda
  I)^{-1})[I+K\widetilde{P}_M(J+F-K\widetilde{P}_M-\lambda I)^{-1}]^{-1}\\
&=I-[I+K\widetilde{P}_M(J+F-K\widetilde{P}_M-\lambda I)^{-1}]^{-1}\,.
\end{align*}
Thus,
\begin{align*}
  \norm{[I+K\widetilde{P}_M(J+F-K\widetilde{P}_M-\lambda I)^{-1}]^{-1}}_{B(\mathcal{H})}\!\!\!&=
  \norm{I-K\widetilde{P}_M\Phi^{-1}(J-\lambda I)^{-1}\Phi}_{B(\mathcal{H})}\\
&
\hspace{-5.2cm}
\le 1+\norm{K}_{B(\mathcal{H})}\norm{\widetilde{P}_M[\Phi(N)]^{-1}}_{B(\mathcal{H})}\norm{(J-\lambda
  I)^{-1}}_{B(\mathcal{H})}\norm{\Phi(N)}_{B(\mathcal{H})}\\
&
\hspace{-5.2cm}
\le 1+\norm{K}_{B(\mathcal{H})}\norm{\widetilde{P}_M[\Phi(M)]^{-1}}_{B(\mathcal{H})}\norm{(J-\lambda
  I)^{-1}}_{B(\mathcal{H})}\,,
\end{align*}
where in the last inequality, we have chosen $N>M$ and taken into
account that $\Phi(N)$ is a contraction. This estimate,
together with \eqref{eq:phi-operator-phi},
\eqref{eq:estimate-resolvent-aux}, and \eqref{eq:aux-inversion}, allow
us to write
\begin{align}
  \norm{\Phi^{-1}(N)(J-\lambda I)^{-1}\Phi(N)}_{B(\mathcal{H})}
\nonumber\\
&
\hspace{-3.8cm}
\le
 \frac{2}{\epsilon(b-\Re\lambda)}\left(1+\norm{K}_{B(\mathcal{H})}
\norm{\widetilde{P}_M[\Phi(M)]^{-1}}_{B(\mathcal{H})}\norm{(J-\lambda
  I)^{-1}}_{B(\mathcal{H})}\right)\nonumber\\
&
\hspace{-3.8cm}
=:C\,,\label{eq:bounded-by-constant}
\end{align}
where the $C$ does not depend on $N$. In view of
Definition~\ref{def:green-function} and (\ref{eq:def-Phi-m}), the last
inequality implies
\begin{equation}
 \label{eq:final-inequality-first}
  \norm{\exp\left(\gamma\sum_{m=1}^{j-1}\phi_\delta(\norm{A_m}_{B(\mathcal{K})})\right)
G_{jk}(\lambda)
\exp\left(-\gamma\sum_{m=1}^{k-1}\phi_\delta(\norm{A_m}_{B(\mathcal{K})})\right)}_{B(\mathcal{K})}
\le C
\end{equation}
for $j,m\le N$. The estimate of the theorem is finally proven by
combining both scalar exponential factors in
(\ref{eq:final-inequality-first}) and letting $N\to\infty$. Formally,
in this proof, $j\ge k$, but the other case is also covered by
recurring to \eqref{eq:norm-resolvent-adjoint}.
\end{proof}
\begin{remark}
  \label{rem:qualified-estimate}
  It is possible to obtain a qualified estimate of the
  constant $C$.
 Indeed, it follows
from (\ref{eq:bounded-by-constant}) that
\begin{equation}
\label{eq:qualified-estimate-pre}
\begin{split}
  &\norm{\exp(\gamma(\lambda)\!\! \sum\limits_{k=\min(m,j)}^{\max(m,j)-1}
\phi_\delta(\norm{A_k}))G_{mj}(\lambda)}\\ &\le
\frac{2}{\epsilon(b-\Re\lambda)}
\left(1+\frac{\abs{b-\min\sigma_p(J)}}{\dist(\lambda,\sigma(J))}
\norm{\widetilde{P}_M[\Phi(M)]^{-1}}_{B(\mathcal{H})}\right)
\end{split}
\end{equation}
for any $j, k\in\nats$ and some $M\in\nats$. Since
\begin{equation*}
 \norm{\widetilde{P}_M[\Phi(M)]^{-1}}_{B(\mathcal{H})} =\norm{\exp(\gamma(\lambda)\sum\limits_{k=1}^{M}
\phi_\delta(\norm{A_k}))}\le e^{\gamma(\lambda) M \delta^{-1/2}}\,,
\end{equation*}
one obtains from (\ref{eq:qualified-estimate-pre}) that
\begin{equation*}
  \norm{\exp(\gamma(\lambda)\!\! \sum\limits_{k=\min(m,j)}^{\max(m,j)-1}
\phi_\delta(\norm{A_k}))G_{mj}(\lambda)}\le
\frac{2\left(1+\frac{\abs{b-\min\sigma_p(J)}}{\dist(\lambda,\sigma(J))}
e^{\gamma(\lambda) M \delta^{-1/2}}\right)}{\epsilon(b-\Re\lambda)}\,.
\end{equation*}

Note that the choice of $M$ is given by
(\ref{eq:estimate-k-projected}) and is independent of $m$ and $j$.
Admissible values of $M$ can be found by considering the canonical
representation of the compact operator $K$ (see
\cite[Chap.\,11, Sec.\,1]{MR1192782}). Let $\{\lambda_k\}_{k}$ be the
sequence of eigenvalues of $J$ which are less than $b$ enumerated so that their multiplicity is taken
into account. This sequence may be finite. Define
\begin{equation*}
  s_k=
  \begin{cases}
    \lambda_k-b & \lambda_k< b\\
    0 &\text{otherwise}\,.
  \end{cases}
\end{equation*}
Thus, for $f\in\mathcal{H}$,
\begin{equation*}
  \norm{K(I-\widetilde{P}_M)f}^2=
\sum_{k=1}^\infty s_k^2\abs{\inner{f}{(I-\widetilde{P}_M)\phi_k}}^2\,,
\end{equation*}
where $\{\phi_k\}_{k=1}^\infty$ is an orthonormal system of
eigenvectors of $J$ corresponding to its spectrum in $(-\infty,
d)$. Therefore, in view of (\ref{eq:estimate-k-projected}), one
arrives at the following condition for $M$.
\begin{equation*}
  \sup_{\norm{f}=1}\sqrt{\sum_k(\lambda_k-b)^2
\abs{\inner{f}{(I-\widetilde{P}_M)\phi_k}}^2}\le
(b-\Re\lambda)\frac{\epsilon}{2}\,.
\end{equation*}
\end{remark}
\begin{lemma}
  \label{lem:perturbed-not-in-spectrum}
  Let $J=J^*$ be the operator given in Definition~\ref{def:j-nought} and
  $L$ be a compact operator in $\mathcal{K}$ with trivial kernel such that
  $\norm{L}_{B(\mathcal{K})}=1$. If $A_m$ has trivial kernel for all
  $m\in\nats$ and $\lambda$ is in the discrete spectrum of $J$, then,
  for any  $\tau>0$ sufficiently small, $\lambda$ is not in the spectrum of
  \begin{equation}
    \label{eq:t-epsilon}
    J(\tau):=J+\tau P_1L^*LP_1
  \end{equation}
(see Notation~(\ref{not-4})).
\end{lemma}
\begin{proof}
  We prove the assertion by \emph{reductio ad absurdum}. Note that,
  since $P_1L^*LP_1$ is a compact operator in $\mathcal{H}$, Weyl
  perturbation theorem \cite[Chp.\,4, Thm.\,5.35]{MR0407617} tells us
  that $\lambda$ is not in the essential spectrum of $J(\tau)$. Suppose
  that for any neighborhood of zero there is $\tau>0$ such that
  $\ker(J(\tau)-\lambda I)$ is not trivial. Pick a nonzero vector
  \begin{equation}
  v_\tau\in\ker(J(\tau)-\lambda I)\label{eq:v_eps-lambda-i}
\end{equation}
If $E$ is the spectral measure of $J$, one has
  \begin{equation}
    \label{eq:one-star}
    E(\{\lambda\})P_1 L^*LP_1 v_\tau=0\quad\text{ for all } \tau>0\,.
  \end{equation}
This is seen by applying the projector $E(\{\lambda\})$ to the equality
\begin{equation}
 \label{eq:t-lambda}
  (J-\lambda I)v_\tau=-\tau P_1L^*LP_1v_\tau
\end{equation}
which is obtained from (\ref{eq:t-epsilon}) and
(\ref{eq:v_eps-lambda-i}). Now, (\ref{eq:t-lambda}) implies in turn that
\begin{equation*}
  E(\{\lambda\}^\complement)(J -\lambda
  I)E(\{\lambda\}^\complement)
 [E(\{\lambda\}^\complement)v_\tau]=-\tau
 E(\{\lambda\}^\complement) P_1L^*LP_1v_\tau\,.
\end{equation*}
where
\begin{equation}
 \label{eq:projector-complement}
  E(\{\lambda\}^\complement):=E(\reals\setminus\{\lambda\})=I-E(\{\lambda\})\,.
\end{equation}
Therefore
\begin{equation*}
  E(\{\lambda\}^\complement)v_\tau=-\tau E(\{\lambda\}^\complement)
 \left[(J -\lambda
   I)\upharpoonright_{\ran(E(\{\lambda\}^\complement))}\right]^{-1}
E(\{\lambda\}^\complement)P_1L^*LP_1v_\tau\,.
\end{equation*}
Taking into account (\ref{eq:projector-complement}),
one obtains from the previous equality that
\begin{equation}
\label{eq:operator-in-braces}
 \begin{split}
  &\left\{I+\tau\left[(J -\lambda
   I)\upharpoonright_{\ran(E(\{\lambda\}^\complement))}\right]^{-1}
  E(\{\lambda\}^\complement)P_1L^*LP_1
\right\}E(\{\lambda\}^\complement)v_\tau\\
&=-\tau\left[(J -\lambda
   I)\upharpoonright_{\ran(E(\{\lambda\}^\complement))}\right]^{-1}
  E(\{\lambda\}^\complement)P_1L^*LP_1 E(\{\lambda\})v_\tau\,.
\end{split}
\end{equation}
Using the fact that
\begin{equation*}
  \norm{\left[(J -\lambda
   I)\upharpoonright_{\ran(E(\{\lambda\}^\complement))}\right]^{-1}}_{B(\mathcal{H})}=
 \frac{1}{\dist(\lambda, \sigma(J)\setminus\{\lambda\})}
\end{equation*}
and choosing
\begin{equation}
 \label{eq:tau-small}
  \tau<\frac{\dist(\lambda, \sigma(J)\setminus\{\lambda\})}{2}\,,
\end{equation}
one verifies that the operator in braces in
(\ref{eq:operator-in-braces}) is invertible and the following estimate holds
\begin{equation}
  \label{eq:estimate-two-stars}
\norm{E(\{\lambda\}^\complement)v_\tau}_{B(\mathcal{H})}\le\frac{2\tau}{\dist(\lambda, \sigma(J)\setminus\{\lambda\})}\norm{LP_1 E(\{\lambda\})v_\tau}_{B(\mathcal{H})}\,.
\end{equation}
Now, it follows from (\ref{eq:one-star}) and
(\ref{eq:projector-complement}) that
\begin{equation*}
  E(\{\lambda\})P_1 L^*LP_1 E(\{\lambda\})v_\tau
  =-E(\{\lambda\})P_1 L^*LP_1 E(\{\lambda\}^\complement)v_\tau\,,
\end{equation*}
whence
\begin{align*}
  \norm{E(\{\lambda\})P_1 L^*LP_1 E(\{\lambda\})v_\tau}&\le
  \norm{E(\{\lambda\}^\complement)v_\tau}\\
 &\le \frac{2\tau}{\dist(\lambda, \sigma(J)\setminus\{\lambda\})}
\norm{LP_1 E(\{\lambda\})v_\tau}\,,
\end{align*}
where, in the last inequality, we have used
(\ref{eq:estimate-two-stars}). If one defines
\begin{equation}
  \label{eq:def-t}
  T:=LP_1
E(\{\lambda\})\,,
\end{equation}
then the inequality above can be written as
\begin{equation*}
  \norm{T^*Tv_\tau}\le
\frac{2\tau}{\dist(\lambda, \sigma(J)\setminus\{\lambda\})}
\norm{Tv_\tau}
\end{equation*}
which in turn implies that
\begin{equation}
  \label{eq:quadratic-form}
  \inner{\left[(T^*T)^2-\frac{4\tau^2}{\dist(\lambda,
\sigma(J)\setminus\{\lambda\})^2}T^*T\right]
v_\tau}{v_\tau}\le 0\,.
\end{equation}

We now show that $T$ is not the zero operator. Indeed, if
$T=0$, then it follows from \eqref{eq:def-t} that, for any $u\in\mathcal{H}$,
\begin{equation}
  \label{eq:projector-e-null}
  P_1E(\{\lambda\})u=0
\end{equation}
due to the fact that $\ker(L)=\{0\}$. But, since the vector
$E(\{\lambda\})u$ is in the kernel of $J-\lambda I$, it should satisfy
the equation (see \eqref{eq:difference-expr})
\begin{align*}
  (\Upsilon E(\{\lambda\}u)_1&=\lambda P_1E(\{\lambda\})u\\
&=0\,,
\end{align*}
where we have used
\eqref{eq:projector-e-null} in the second equality. This
last expression implies, via \eqref{eq:difference-initial}, that
\begin{equation*}
B_1 P_1E(\{\lambda\}u + A_1 P_2E(\{\lambda\}u=0
\end{equation*}
and from this, using again \eqref{eq:projector-e-null}, one obtain
$A_1 P_2E(\{\lambda\}u=0$ which in turn implies that $
P_2E(\{\lambda\}u=0$ since $\ker(A_1)$ is trivial. Having
established that $P_1E(\{\lambda\}u=0$ and $P_2E(\{\lambda\}u=0$, one
finds recurrently from \eqref{eq:difference-recurrence}, taking into
account that $\ker(A_m)=\{0\}$ for all $m\in\nats$, that
\begin{equation*}
  P_3E(\{\lambda\}u=P_4E(\{\lambda\}u=\dots=0\,.
\end{equation*}
Therefore $E(\{\lambda\})u=0$ for any $u\in\mathcal{H}$ which is a
contradiction since $\lambda$ is in the discrete spectrum of $J$. Thus
we have shown that $T\ne 0$ and then $T^*T\ne 0$.

Now, since $T^*T$ is a compact nonzero operator, one can take $\tau$
such that, in addition to (\ref{eq:tau-small}), satisfies
\begin{equation*}
  \frac{4\tau^2}{\dist(\lambda, \sigma(J)\setminus\{\lambda\})^2}
<\min\sigma(T^*T)\setminus\{0\}\,.
\end{equation*}
By our choice of $\tau$, the quadratic form in
(\ref{eq:quadratic-form}) cannot be negative. Indeed, if
$\sigma(T^*T)\setminus\{0\}=\{\mu_k\}_{k=1}^p$ ($p\in\nats$ since $\rank(T^*T)$ is
finite due to the fact that $\lambda\in\sigma_{\rm discr}(J)$), then
\begin{equation*}
  \mu_k^2-\frac{4\tau^2}{\dist(\lambda, \sigma(J)\setminus\{\lambda\})^2}\mu_k>0
\end{equation*}
for all $k=1,\dots,p$ in view of the fact that
$\min\sigma(T^*T)\setminus\{0\}>0$. This conclusion and
(\ref{eq:quadratic-form}) imply that
\begin{equation*}
  v_\tau\in\ker\left[(T^*T)^2-\frac{4\tau^2}{\dist(\lambda,
\sigma(J)\setminus\{\lambda\})^2}T^*T\right]\,.
\end{equation*}
Therefore $v_\tau\in\ker(T^*T)=\ker(T)$ which yields, recalling \eqref{eq:def-t}
and the fact that $\ker(L)=\{0\}$, that
$P_1v_\tau=0$.

We now show that $P_1v_\tau=0$ leads to $v_\tau=0$ which is a
contradiction to our assumption. Since $v_\tau$ is an eigenvector of
$J(\tau)$ at $\lambda$, it follows from \eqref{eq:difference-initial}
that
\begin{equation*}
  (B_1 +\tau L^*L)P_1v_\tau + A_1P_2v_\tau=\lambda P_1v_\tau
\end{equation*}
and therefore $P_2v_\tau=0$ since the kernel of $A_1$ is trivial.  As
before, having established that $P_1v_\tau,P_2v_\tau=0$, it follows by
iteration of \eqref{eq:difference-recurrence} that
\begin{equation*}
   P_3v_\tau=P_4v_\tau=\dots=0
\end{equation*}
since $\ker(A_m)=\{0\}$ for all $m\in\nats$.
\end{proof}
\begin{theorem}
\label{thm:infinite-interval-in-spectrum}
Assume that the operator $J$ given in Definition~\ref{def:j-nought} is
self-adjoint and bounded from below and $\ker(A_m)=\{0\}$ for all $m\in\nats$. Consider
real numbers $b$ and $\lambda$ such that $(-\infty,b)\cap\sigma_{ess}(J)=\emptyset$
and $\lambda\in(-\infty,b)\cap\sigma_p(J)$.  Fix
$\delta>0$ and $\epsilon\in(0,1)$.  If $u$
is an eigenvector corresponding to $\lambda$, normalized so that $\norm{u}=1$,
then
  \begin{equation*}
    \norm{u_m}_{\mathcal{K}}\le
    C\exp(-\gamma(\lambda)\sum_{k=1}^{m-1}
\phi_\delta(\norm{A_k}_{B(\mathcal{K})}))\,.
  \end{equation*}
  where $C$ does not depend on $m$ and the functions $\phi_\delta$
  and $\gamma$ are given in \eqref{eq:phi-function-def}
  and (\ref{eq:def-gamma}), respectively.
\end{theorem}
\begin{proof}
  By Lemma~\ref{lem:perturbed-not-in-spectrum}, one can choose
  $\tau>0$ such that $\lambda\not\in\sigma(J(\tau))$. According to perturbation
  theory, $J(\tau)$ is bounded from below and
  $\sigma_{ess}(J)=\sigma_{ess}(J(\tau))$ (see \cite[Chp.\,4,
  Thm.\,5.35 and Chp.\,5, Thm.\,4.11]{MR0407617}). Thus, the estimates of
  Theorem~\ref{thm:infinite-interval} can be applied to $J(\tau)$.

  If $u$ is a nonzero vector in $\ker(J-\lambda I)$, then
  \begin{equation*}
    (J(\tau)-\lambda I)u=\tau P_1L^*LP_1u\,.
  \end{equation*}
Therefore,
\begin{equation}
\label{eq:eigenvector-throu-itself}
   u=\tau(J(\tau)-\lambda I)^{-1}P_1L^*LP_1u
\end{equation}
which in turn implies
\begin{align*}
  \norm{u_m}_{\mathcal{K}}&=\norm{P_mu}_{\mathcal{H}}\\ &=\norm{P_m\tau(J(\tau)-\lambda
                             I)^{-1}P_1L^*LP_1u}_{\mathcal{H}}\\
&\le\tau\norm{P_m(J(\tau)-\lambda I)^{-1}P_1}_{B(\mathcal{H})}
\norm{P_1u}_\mathcal{H}\\
&\le C\exp(-\gamma(\lambda)\sum_{k=1}^{m-1}
\phi_\delta(\norm{A_k}_{B(\mathcal{K})}))\norm{P_1u}_\mathcal{H}\,,
\end{align*}
where in the first inequality we use that
$\norm{L}_{B(\mathcal{K})}=1$ and in the second we resort to
Theorem~\ref{thm:infinite-interval}.
\end{proof}
\begin{remark}
  \label{rem:semi-bounded}
  The semi-boundedness of $J$ is essential. Indeed, there are examples
  of (scalar) Jacobi operators \cite{MR2579689}, where accumulation of
  $\sigma_{\rm discr}(J)$ at infinity is allowed, whose generalized eigenvalues
  has an asymptotic behavior depending on the main diagonal entries of
  the Jacobi operator. The estimates given in the previous theorem
  show that, in the semi-bounded case, the diagonal block entries do not
  play any role in the component-wise estimates of the generalized
  eigenvectors.
\end{remark}
Under additional conditions on the asymptotic behaviour of the
sequence $\{\norm{A_k}\}_{k=1}^\infty$, it is possible to simplify the
expression for the estimates in Theorems~\ref{thm:infinite-interval}
and \ref{thm:infinite-interval-in-spectrum}. This is done in the
following assertion.

\begin{corollary}
  \label{cor:norm-to-infty}
  Assume that the operator $J$ given in Definition~\ref{def:j-nought}
  is self-adjoint and bounded from below and that
  $\norm{A_k}\convergesto{k}\infty$. Let the real number $b$ be such
  that $(-\infty,b)\cap\sigma_{ess}(J)=\emptyset$. Fix $\epsilon$
  arbitrarily small in $(0,1)$.
\begin{enumerate}[(a)]
\item \label{not-in-spectrum-cor}
If $\lambda\not\in\sigma(J)$ and $\Re\lambda<b$, then
  \begin{equation*}
       \norm{G_{mj}(\lambda)}_{B(\mathcal{K})}\le
    C_a
   \exp(-(1-\epsilon)\sqrt{b-\Re\lambda}\!
  \sum\limits_{k=\min(m,j)}^{\max(m,j)-1}1/\sqrt{\norm{A_k}_{B(\mathcal{K})}})\,.
  \end{equation*}
\item \label{in-spectrum-cor}
Under the assumption that $\ker(A_k)=\{0\}$ for all
  $k\in\nats$, if $\lambda\in\sigma_p(J)\cap (-\infty,b)$ and $u$ is the corresponding
eigenvector, normalized so that $\norm{u}=1$,
then
  \begin{equation*}
    \norm{u_m}_{\mathcal{K}}\le
   C_b\exp(-(1-\epsilon)\sqrt{b-\lambda}
\sum_{k=1}^{m-1}1/\sqrt{\norm{A_k}_{B(\mathcal{K})}})\,.
  \end{equation*}
  \end{enumerate}
The constant $C_a$ does not
depend on $m$ and $j$, and $C_b$ does not depend on $m$.
\end{corollary}
\begin{proof}
  By choosing $\delta$ appropriately (essentially sufficiently large),
  one obtains
  \begin{equation*}
    \frac{(b-\Re\lambda)(1-\epsilon)}{\delta}<\epsilon_1\ll 1\,.
  \end{equation*}
 Given $b$ and $\epsilon$, the choice of $\delta$ depends on
 $\lambda$.

The presence of the factor $(1-\epsilon)\sqrt{b-\Re\lambda}$ instead
of $\gamma(\lambda)$ follows from the fact that, if $0<t<\epsilon_1$, then
\begin{equation}
 \label{eq:correct-inequality-psi-}
  \psi^{-1}(t)\ge\sqrt{t}(1-\epsilon_2)\,,
\end{equation}
where $\epsilon_2$ is arbitrarily small whenever $\epsilon_1$
is sufficiently small. Indeed, if
(\ref{eq:correct-inequality-psi-}) holds, then the choice of $\gamma$
in \eqref{eq:def-gamma} can be replaced by
\begin{equation*}
  \gamma=\sqrt{\delta}\sqrt{\frac{(b-\Re\lambda)(1-\epsilon)}{\delta}}
 (1-\epsilon_2)=\sqrt{(b-\Re\lambda)(1-\widetilde{\epsilon})}\,,
\end{equation*}
where $\widetilde{\epsilon}$ is arbitrarily small.

Let us show that (\ref{eq:correct-inequality-psi-}) holds. From
\eqref{eq:psi-function-def}, if $t=ye^{\sqrt{y}}$, one has
\begin{equation*}
  \psi^{-1}(t)=\sqrt{t}e^{-\frac{1}{2}\sqrt{y}}=
  \sqrt{t}\exp(-\frac{1}{2}\sqrt{t}e^{-\frac{1}{2}\sqrt{y}})
\ge \sqrt{t}e^{-\frac12\sqrt{t}}\,.
\end{equation*}
Thus, if $t<\epsilon_1$, then
\begin{equation*}
  \psi^{-1}(t)\ge\sqrt{t}e^{-\frac12\sqrt{\epsilon_1}}=\sqrt{t}(1-\epsilon_2)\,.
\end{equation*}
Finally, observe that
\begin{align*}
  \sum_{k=1}^m\phi_\delta(\norm{A_k})&=
 \sum_{\substack{\left\| A_k\right\|\le\delta\\ k
\le m}}\frac{1}{\sqrt{\delta}}+
\sum_{\substack{\left\| A_k\right\|>\delta\\ k
\le m}}\frac{1}{\sqrt{\norm{A_k}}}\\
&=\sum_{k=1}^m \frac{1}{\sqrt{\norm{A_k}}} +
 \sum_{\substack{\left\| A_k\right\|\le\delta\\ k
\le m}}\frac{1}{\sqrt{\delta}}-
\sum_{\substack{\left\| A_k\right\|\le\delta\\ k
\le m}}\frac{1}{\sqrt{\norm{A_k}}}
\end{align*}
where the second and third terms can be absorbed into the constant
which does not depend on $m\gg 1$ since $\norm{A_k}\convergesto{k}\infty$.
\end{proof}

\begin{remark}
  \label{rem:sharpness-1}
  The choice of the function $\gamma$ in
  Theorems~\ref{thm:infinite-interval} and
  \ref{thm:infinite-interval-in-spectrum} is rather
  optimal. If one assumes that $b$ is the border of the
  essential spectrum of $J$, then the behavior of $\gamma(\lambda)$,
  when $\lambda$ approaches $b$, is the most interesting in
  applications. Formula (\ref{eq:psi-function-def}) leads to the following
  asymptotic behavior
  \begin{equation}
    \label{eq:gamma-asymptotic}
    \gamma(\lambda)\simeq\sqrt{(b-\Re\lambda)(1-\epsilon)}
  \end{equation}
  as $\lambda\to b$. Note that $\sqrt{1-\epsilon}$ may be chosen arbitrarily
  close to $1$, but it has been shown in the case of scalar Jacobi
  matrices \cite{MR3028179} that the coefficient $(1-\epsilon)$ cannot
  be replaced by any number greater than $1$, although arbitrarily close to
  $1$. This proves the sharpness of the estimates of
  Theorems~\ref{thm:infinite-interval} and
  \ref{thm:infinite-interval-in-spectrum}. Namely, the factor
  $(1-\epsilon)$ in the definiton of $\gamma$  (see (\ref{eq:def-gamma})) cannot be replaced by $(1+\epsilon)$ for any
  arbitrary small $\epsilon>0$.
\end{remark}

\section{Estimates of generalized eigenvectors in the case of
  commuting entries}
\label{sec:case-comm-entr}

If in Theorems~\ref{thm:infinite-interval} and
\ref{thm:infinite-interval-in-spectrum} one additionally requires
that the matrices $A_m,B_m,A_m^*$ commute for
$m\in\nats$, then a refinement of the previous results occurs
so that it is possible to estimate the growth of the generalized
eigenvectors along any spatial ``direction'' in $\mathcal{K}$. The
refined results are given in the next assertion, where Notation~(\ref{not-5}) is used.
We draw the reader's attention to the fact that in this
  section we use functions of operators given by the functional
  calculus of self-adjoint operators.

\begin{theorem}
  \label{thm:infinite-interval-commutation}
Assume that the operator $J$ given in Definition~\ref{def:j-nought}
  is self-adjoint and bounded from below and the system of operators
    \begin{equation*}
      \{A_m,B_m,A^*_m\}_{m\in\nats}
    \end{equation*}
commutes pairwise. Take a real
  number $b$ such that
  $(-\infty,b)\cap\sigma_{ess}(J)=\emptyset$ and consider a complex
  number $\lambda$ with $\Re\lambda<b$. Fix $\delta>0$ and
  $\epsilon$ arbitrarily small in $(0,1)$.
  \begin{enumerate}[(i)]
  \item \label{not-in-spectrum-commutation}
If $\lambda\not\in\sigma(J)$ and $\Re\lambda<b$, then
    \begin{equation}
 \label{eq:estimate-green-commutation}
 \norm{\exp(\gamma(\lambda)\!\!\sum
\limits_{k=\min(m,j)}^{\max(m,j)-1}\phi_\delta(\abs{A_k}))G_{mj}(\lambda)}
\le C\,.
\end{equation}
\item \label{in-spectrum-commutation}
Under the assumption that $\ker(A_k)=\{0\}$ for all
  $k\in\nats$, if $\lambda\in\sigma_p(J)\cap (-\infty,b)$ and $u$ is the corresponding
  eigenvector, normalized so that $\norm{u}_{\mathcal{H}}=1$, then
  \begin{equation}
    \label{eq:estimate-eigen-communtation}
    \norm{\exp(\gamma(\lambda)\sum_{k=1}^{m-1}\phi_\delta(\abs{A_k}))u_m}_{\mathcal{K}}\le \widetilde{C}\,.
  \end{equation}
\end{enumerate}
In both (\ref{not-in-spectrum-commutation}) and
(\ref{in-spectrum-commutation}), $\phi_\delta$ and $\gamma$ are given by
(\ref{eq:phi-function-def}) and (\ref{eq:def-gamma}), respectively. The constant $C$ does not
depend on $m$ and $j$, and $\widetilde{C}$ does not depend on $m$.
\end{theorem}
\begin{proof}
  First, we prove (\ref{not-in-spectrum-commutation}). We consider
  again the operators $J_b$ and $K$ defined in the proof of
  Theorem~\ref{thm:infinite-interval}, but modify the definition of the
  operators $\Phi$. For any fixed $N\in\nats$, define the bounded
  operators on $\mathcal{K}$
  \begin{equation}
    \label{eq:phi-m-new-def}
      \Phi_m:=
\begin{cases}
\exp\left(-\gamma\sum_{k=1}^{m-1}\phi_\delta(\abs{A_k})\right)\,,&
m\le N\,,\\
\exp\left(-\gamma\sum_{k=1}^{N-1}\phi_\delta(\abs{A_k})\right)\,,&
m> N\,,
\end{cases}
  \end{equation}
and the bounded operator on $\mathcal{H}$ by
\begin{equation*}
  \Phi:=\diag\{\Phi_m\}_{m=1}^\infty
\end{equation*}
Similar to what we had in the proof of
Theorem~\ref{thm:infinite-interval}, $\Phi$ depends on $N$ and
$\Phi(N)$ is a boundedly invertible contraction for any finite
$N$. Note that this time the block operator $\Phi_m$ is not a scalar
operator.

Consider the operator $F\in B(\mathcal{H})$ such that
(\ref{eq:relation-f-phi}) is satisfied with our new $\Phi$. Repeating
the argumentation in the proof of Theorem~\ref{thm:infinite-interval},
one arrives at (\ref{eq:real-part-first-estimate}). Using
(\ref{eq:phi-m-new-def}) and the fact that the system
$\{A_m,B_m,A^*_m\}_{m\in\nats}$ commutes, one obtains from
(\ref{eq:real-part-first-estimate}) that
\begin{equation}
 \label{eq:inequality-real-part-modulus-operator}
    \norm{\Re F}_{B(\mathcal{H})}\le
\sup_{m\in\nats}
\left\{\norm{\abs{A_m}\left(e^{-\gamma\phi_{\delta}(\abs{A_m})}
-2I+e^{\gamma\phi_\delta(\abs{A_m})}\right)}_{B(\mathcal{K})}\right\}\,.
\end{equation}
Due to the inequality
\begin{equation*}
  e^X-2 I+e^{-X}\le X^2e^X
\end{equation*}
valid for any positive operator $X$ and obtained from
(\ref{eq:algebraic-inequality-exp}) by the spectral theorem, one derives
from (\ref{eq:inequality-real-part-modulus-operator}) the estimate
\begin{equation*}
   \norm{\Re F}_{B(\mathcal{H})}\le
\sup_{m\in\nats}
\left\{\norm{\abs{A_m}\gamma^2\phi_\delta^2(\abs{A_m})
e^{\gamma\phi_\delta(\abs{A_m})}}\right\}\,.
\end{equation*}
But
\begin{equation*}
  \sup_{m\in\nats}
\left\{\norm{\abs{A_m}\gamma^2\phi_\delta^2(\abs{A_m})
e^{\gamma\phi_\delta(\abs{A_m})}}\right\}\le
 \gamma^2e^{\gamma/\sqrt{\delta}}
\end{equation*}
since, again by the spectral theorem and the definition of
$\phi_\delta$ given in \eqref{eq:phi-function-def}, one has
\begin{equation*}
  \abs{A_m}\phi_\delta^2(\abs{A_m})\le I\,.
\end{equation*}
Following the reasoning of the proof of
Theorem~\ref{thm:infinite-interval}, one verifies that
(\ref{eq:realpart-c}) holds as long as $\gamma$ is given by
(\ref{eq:def-gamma}).  The rest of the proof repeats the one of
Theorem~\ref{thm:infinite-interval} up to
(\ref{eq:bounded-by-constant}) from which, in view of
(\ref{eq:phi-m-new-def}), one obtains
\begin{equation*}
    \norm{\exp\left(\gamma\sum_{m=1}^{j-1}\phi_\delta(\abs{A_m})\right)
G_{jk}(\lambda)
\exp\left(-\gamma\sum_{m=1}^{k-1}\phi_\delta(\abs{A_m})\right)}_{B(\mathcal{K})}
\le C\,.
\end{equation*}
The assertion (\ref{not-in-spectrum-commutation}) follows from this
inequality by combining the operators on both sides of $G_{jk}$ and
letting $N\to\infty$. In this proof, $j\ge k$, but the other case is
also covered by recurring to \eqref{eq:norm-resolvent-adjoint}.

To prove (\ref{in-spectrum-commutation}), one resorts to
Lemma~\ref{lem:perturbed-not-in-spectrum} and choose $\tau>0$ so that
$\lambda\not\in\sigma(J(\tau))$. As in the proof of
Theorem~\ref{thm:infinite-interval-in-spectrum}, it follows from
(\ref{eq:t-epsilon}) that if $u$ is in $\ker(J-\lambda I)$, then
(\ref{eq:eigenvector-throu-itself}) holds. Hence
\begin{equation*}
  \norm{u_m}_{\mathcal{K}}
\le\tau\norm{P_m(J(\tau)-\lambda I)^{-1}P_1}_{B(\mathcal{H})}
\norm{P_1u}_\mathcal{H}\,.
\end{equation*}
For finishing the proof, it only remains to note that $J(\tau)$
satisfies the hypothesis of (\ref{not-in-spectrum-commutation}).
\end{proof}
\begin{remark}
  \label{rem:use-of-qualified-estimate}
  Qualified estimates of the constants
    $C$ and $\widetilde{C}$ given in
    \eqref{eq:estimate-green-commutation} and
    \eqref{eq:estimate-eigen-communtation}, respectively, can be
    obtained by following the reasoning of
    Remark~\ref{rem:qualified-estimate} with the scalar
    $\phi_\delta(\norm{A_k})$ substituted by the operator
    $\phi_\delta(\abs{A_k})$ for all $k\in\nats$. Note that
the operator
\begin{equation*}
\exp(\gamma(\lambda)\!\!\sum\limits_{k=\min(m,j)}^{\max(m,j)-1}
\phi_\delta(\abs{A_k}))
\end{equation*}
in (\ref{eq:estimate-green-commutation}) and
(\ref{eq:estimate-eigen-communtation}) governs the growth of the
generalized eigenvectors.
\end{remark}
\begin{corollary}
  \label{cor:norm-to-infty-commutation}
  Let $J$ be the operator given in Definition~\ref{def:j-nought} such
  that it is self-adjoint, bounded from below, and the operators $A_m$ and $B_m$ satisfy the conditions of
  Theorem~\ref{thm:infinite-interval-commutation} for all
  $m\in\nats$. Assume, additionally that
  $\norm{A_m^{-1}}\convergesto{m}0$. Fix an $\epsilon\in(0,1)$.
\begin{enumerate}[a)]
\item \label{not-in-spectrum-cor-commutation}
If $\lambda\not\in\sigma(J)$, then
  \begin{equation*}
       \norm{\exp((1-\epsilon)\sqrt{b-\Re\lambda}\!
  \sum\limits_{k=\min(m,j)}^{\max(m,j)-1}1/\sqrt{\abs{A_k}})G_{mj}(\lambda)}\le
    C_a\,.
  \end{equation*}
\item \label{in-spectrum-cor-commutation}
If $\lambda\in\sigma_p(J)$ and $u$ is the corresponding
  eigenvector, normalized so that $\norm{u}_{\mathcal{H}}=1$, then
  \begin{equation*}
    \norm{\exp((1-\epsilon)\sqrt{b-\Re\lambda}
\sum_{k=1}^{m-1}1/\sqrt{\abs{A_k}})u_m}_{\mathcal{K}}\le
   C_b\,.
  \end{equation*}
  \end{enumerate}
The constant $C_a$ does not
depend on $m$ and $j$, and $C_b$ does not depend on $m$.
\end{corollary}
\begin{proof}
  We prove the claim in (b). The
  assertion (\ref{not-in-spectrum-cor-commutation}) is proven
  analogously.  We repeat part of the argumentation of the proof of
  Corollary~\ref{cor:norm-to-infty}. After having shown that
  $\gamma(\lambda)$ can be substituted by
  $(1-\epsilon)\sqrt{b-\Re\lambda}$, one arrives at
  \begin{align*}
    &\norm{\exp((1-\epsilon)\sqrt{b-\Re\lambda}
\sum_{k=1}^{m-1}\phi_\delta(\abs{A_k}))u_m}_{\mathcal{K}}\\ &\le
\norm{\exp\left[(1-\epsilon)\sqrt{b-\Re\lambda}\left(
\sum_{k=1}^{m-1}(\abs{A_k})^{-\frac{1}{2}} +
\sum_{\substack{\abs{A_k}\le\delta I\\ k
\le m}}\left(\frac{1}{\sqrt{\delta}}I -
(\abs{A_k})^{-\frac{1}{2}}
\right)\right)\right]u_m}_{\mathcal{K}}\\
&\le \widetilde{C}
\norm{\exp((1-\epsilon)\sqrt{b-\Re\lambda}
\sum_{k=1}^{m-1}(\abs{A_k})^{-\frac{1}{2}})u_m}_{\mathcal{K}}\,,
  \end{align*}
where in passing to the last inequality, we have used the pairwise
commutativity of the elements of the sequence
$\{A_k\}_{k=1}^\infty$. Note also that, under the assumption that $m$
is sufficiently large, the operator
\begin{equation*}
  \sum_{\substack{\abs{A_k}\le\delta I\\ k
\le m}}\left(\frac{1}{\sqrt{\delta}}I -
(\abs{A_k})^{-\frac{1}{2}}
\right)
\end{equation*}
is uniformly (with respect to $m$) bounded.
\end{proof}

\begin{remark}
  \label{rem:sharpness-2}
  The form of the estimates given in
  Theorem~\ref{thm:infinite-interval-commutation} are optimal. Taking
  into consideration that $\phi_\delta(x)=1/\sqrt{x}$ when $x>\delta$
  and $\gamma(\lambda)\simeq \sqrt{(b-\Re\lambda)(1-\epsilon)}$ as
  $\lambda$ approaches $b$, one proves the estimate sharpness as in
  the scalar case $d=1$ repeating the reasoning given in
  Remark~\ref{rem:sharpness-1}. However the sharpness of Theorem
  \ref{thm:infinite-interval-commutation} has a deeper character even
  in the ``trivial'' case where our block Jacobi matrix $J$ is the
  orthogonal sum of $d$ different copies of scalar Jacobi matrices
  (all entries are diagonal matrices). Indeed, in this case,
  Theorem~\ref{thm:infinite-interval-commutation} provides us with a
  sharp estimate for each scalar copy separately. Note that in the
  commuting case the reduction of $J$ to the orthogonal sum of $d$
  copies of scalar Jacobi matrices generally cannot be performed
  effectively. A possible exception is the very special case when the matrix
  entries are scalar proportional to some fixed commuting matrices.
\end{remark}

\section{An example with noncommuting entries}
\label{sec:an-example-noncommuting}
This example corresponds to the case $d=2$. For $\alpha \in (0,1)$, define
\begin{equation*}
  A_n=A_n^*:=
  \begin{pmatrix}
    0&r_n\\
    r_n&0
  \end{pmatrix}\,,\qquad r_n=n^\alpha\,,
\end{equation*}
and
\begin{equation*}
  B_n:=
  \begin{pmatrix}
    s_n&0\\
    0&t_n
  \end{pmatrix}\,,\qquad s_n=sn^\alpha\,, \quad t_n=tn^\alpha\,,\quad s,t>0\,.
\end{equation*}
By the moment, we consider arbitrary values for $s$ and $t$. Later, we
will impose extra conditions on them.

Note that the block Jacobi operator $J$ whose
matrix representation is \eqref{eq:block-jm} with the entries given
above cannot be reduced to a orthogonal sum of scalar Jacobi matrices
because $A_n$, $B_n$ do not commute if $s\ne t$.
Let
\begin{equation*}
  \mathcal{B}_n:=
  \begin{pmatrix}
    0& I \\
    -A_n^{-1}A_{n-1}&A_n^{-1}(\lambda I_2-B_n)
  \end{pmatrix}
\,,\quad\lambda\in\reals\,,
\end{equation*}
be the transfer matrix of $J$ associated with sequences $A_n$,
$B_n$. Here $I$ is the unit matrix in $\mathcal{K}=\complex^2$ (recall
the notation given in Section~\ref{sec:block-jacobi}). For the
spectral analysis of $J$ we need to find asymptotic formulae for
the eigenvalues of $\mathcal{B}_n$. First recall that a necessary and
sufficient condition for the invertibility of the $4\times 4$ matrix
$\mathcal{B}_n-\mu I$ is the invertibility of the $2\times 2$ matrix
\begin{equation*}
  X_n:=A_n^{-1}(\lambda I -B_n)-\mu I-\left[-A_n^{-1}A_{n-1}(-\mu^{-1}I)\right]\,.
\end{equation*}
This is clear if one uses the Schur-Frobenius complement. Multiplying $X_n$
by $\mu$, one has
\begin{equation*}
  \mu X_n=
  \begin{pmatrix}
    -\mu^2-1+O(\frac{1}{n})& \mu(\frac{\lambda}{n^\alpha}-t)\\
    \mu(\frac{\lambda}{n^\alpha}-s)& -\mu^2-1+O(\frac{1}{n})
  \end{pmatrix}
\end{equation*}
for $n$ sufficiently large. Moreover, the diagonal elements of $\mu
X_n$ are equal and $\det(\mu X_n)$ vanishes if and only if
\begin{equation*}
  (-\mu^2-1+O(n^{-1}))^2-
  \mu^2\left(\frac{\lambda}{n^\alpha}-
    t\right)\left(\frac{\lambda}{n^\alpha}-s\right)
  =0\,.
\end{equation*}
The last equation is equivalent to
\begin{equation*}
  \mu^2+1+O(n^{-1})=\pm\mu\sqrt{\left(\frac{\lambda}{n^\alpha}-
    t\right)\left(\frac{\lambda}{n^\alpha}-s\right)}\,,
\end{equation*}
which yields the following four eigenvalues of $\mathcal{B}_n$
corresponding to the four possible choices of signs $+$, $-$ below.
\begin{equation*}
  \mu_n=\mp\frac12\sqrt{\left(\frac{\lambda}{n^\alpha}-
      t\right)\left(\frac{\lambda}{n^\alpha}-s\right)}\pm
  \sqrt{\frac{1}{4}\left(\frac{\lambda}{n^\alpha}-
    t\right)\left(\frac{\lambda}{n^\alpha}-s\right)-1+O\left(\frac1n\right)}\,.
\end{equation*}
Note that the $O(n^{-1})$ terms are all real.  If one chooses
$\alpha\in(\frac12,1)$, then, in the special case $st=4$, the last
formula can be written as
\begin{align}
  \label{eq:eigenvalues-transfer-matrix}
  \mu_n&=\mp\frac12\sqrt{\left(\frac{\lambda}{n^\alpha}-
         t\right)\left(\frac{\lambda}{n^\alpha}-s\right)}\pm
         \sqrt{-\frac{\lambda(s+t)}{4n^\alpha}+
         O\left(\frac{1}{n}\right)}\nonumber\\
  &=\mp\frac12\sqrt{\left(\frac{\lambda}{n^\alpha}-
    t\right)\left(\frac{\lambda}{n^\alpha}-s\right)}\pm
    \frac{i}{2}n^{-\alpha/2}\sqrt{\lambda}\sqrt{s+t}\left(1+O\left(n^{\alpha-1}\right)\right)
    \nonumber\\
       &=\mp\left[1\pm\frac{i\sqrt{\lambda}}{2}n^{-\alpha/2}\sqrt{s+t}
         -\frac{(s+t)\lambda}{4n^\alpha}+O\left(n^{\alpha/2-1}\right)\right]
\end{align}
This formula will be used to give an estimate of the growth of
generalized eigenvectors of $J$. Note that if $st\ne 4$, then
$\mu_n\to \mp\frac12\sqrt{st}\pm\frac12\sqrt{st-4}$ as $n\to\infty$. Since
this value does not coincide with $\pm 1$, it provides uniformly with
respect to $\lambda$ elliptic (if $st<4$) or hyperbolic (if $st>4$)
behavior of solutions of the formal spectral equation. Therefore the
only case where the value $\lambda$ is essential (producing the unbounded
gap) is the situation where $st=4$.

Now we turn to the proof of the nonnegativity of $J$, modulo compact
operators, i.\,e., the existence of a compact operator $K$ such that
$J+K\ge 0$.
Consider the quadratic form of $J$, viz.,
\begin{equation*}
  (Ju,u)=\sum_{n=1}^\infty(A_{n-1}u_{n-1}+B_nu_n+A_nu_{n+1},u_n)_{\mathcal{K}}\,.
\end{equation*}
Since the vectors $u=\{u_n\}_{n=1}^\infty$ with finitely many nonzero
elements form a core for $J$, for calculating the quadratic form of
$J$ it suffices to calculate it in such
vectors. Write $u_n=d_nv_n$ where $d_n=n^{-\alpha/2}$. Using the identities
\begin{equation*}
  d_nd_{n-1}(n-1)^\alpha=1+O(n^{-1})\,,\quad
  d_nd_{n+1}n^\alpha
  =1+O(n^{-1})
\end{equation*}
one obtains
\begin{align*}
  (Ju,u)&=\sum_{n=1}^\infty\left(
    \begin{pmatrix}
      0&1\\1&0
    \end{pmatrix}
v_{n-1} +
\begin{pmatrix}
  s&0\\0&t
\end{pmatrix}
v_n+
    \begin{pmatrix}
      0&1\\1&0
    \end{pmatrix}
v_{n+1}, v_n
\right)_{\mathcal{K}}\\
&+
\sum_{n=1}^\infty\left(
O(n^{\alpha-1})
u_{n-1} +
O(n^{\alpha-1})
u_{n+1}, u_n
\right)_{\mathcal{K}}
\,.
\end{align*}
The last series in the equality above corresponds to the Jacobi
operator in $\mathcal{H}=l_2(\nats,\complex^2)$ with the subdiagonals
decaying as $O(n^{\alpha-1})$ with $\alpha<1$. Therefore it defines a
compact operator in $\mathcal{H}$. Hence the problem of positivity has
been reduced to the question of positivity of the block Jacobi matrix
$J_c$ defined by constant entries
\begin{equation*}
  A_n:=
  \begin{pmatrix}
    0&1\\1&0
  \end{pmatrix}
\,,\qquad B_n:=
\begin{pmatrix}
  s&0\\0&t
\end{pmatrix}
\,.
\end{equation*}
We prove the following general result for $s,t>0$.
\begin{lemma}
\label{lem:jc-bounded-below}
If $s,t>0$, then
\begin{equation*}
  J_c\ge\frac{st-4}{\frac{t+s}{2}+\left[(\frac{t-s}{2})^2+4\right]^{1/2}}I\,.
\end{equation*}
\end{lemma}
\begin{proof}
  For $v_n=(f_n,g_n)^{\rm\textsf{T}}$, we have
  \begin{align*}
    (J_c v,v)&=\sum_{n=1}^\infty\left(
      \begin{pmatrix}
        g_{n-1}+sf_n+g_{n+1}\\
        f_{n-1}+tg_n+f_{n+1}
      \end{pmatrix},
\begin{pmatrix}
f_n\\ g_n
\end{pmatrix}
\right)_{\complex^2}\\
 &\ge\sum_{n=1}^\infty
(s\abs{f_n}^2 +
t\abs{g_n}^2 -
2\abs{f_n}\abs{g_{n-1}}-
2\abs{f_{n-1}}\abs{g_n})\\
 &\ge\sum_{n=1}^\infty
(s\abs{f_n}^2 +
t\abs{g_n}^2 -
\frac{1}{\epsilon}\abs{f_n}^2-
\epsilon\abs{g_{n-1}}^2-\eta\abs{g_n}^2-\frac{1}{\eta}\abs{f_{n-1}}^2)
 \end{align*}
for any $\epsilon,\eta>0$. Let us optimize the choice of $\epsilon$ and
$\eta$. Note that the last sum can be written as
\begin{equation*}
  \sum_{n=1}^\infty\left[\left(s-\frac{1}{\epsilon}-\frac{1}{\eta}\right)
\abs{f_n}^2 + \left(t-\epsilon-\eta\right)\abs{g_n}^2\right]\,.
\end{equation*}
Choose $\epsilon,\eta$ so that
\begin{equation}
\label{eq:aux-epsilon-eta-st}
  s-\frac{\epsilon+\eta}{\epsilon\eta}=t-(\epsilon+\eta)\,,
\end{equation}
and define the new variable $k:=\eta\epsilon^{-1}$. Then, for a fixed
$k$, the identity (\ref{eq:aux-epsilon-eta-st}) is equivalent to
\begin{equation*}
  \epsilon^2+(1+k)^{-1}(s-t)\epsilon-\frac{1}{k}=0\,.
\end{equation*}
The positive solution of the last equation is given by
\begin{equation*}
  \epsilon=[2(1+k)]^{-1}(t-s)+ \left[\left\{\frac12(1+k)^{-1}(t-s)\right\}^2+ k^{-1}\right]^{1/2}\,.
\end{equation*}
Since
\begin{equation}
\epsilon\eta=(1+k)\epsilon=\frac{t-s}{2}+
\left[\frac{(t-s)^2}{4}+\frac{(k+1)^2}{k}\right]^{1/2}\,,
\end{equation}
one checks that the minimum of $(1+k)\epsilon$ taken for $k>0$ is
equal to
\begin{equation*}
  \frac{t-s}{2}+
\left[\left(\frac{t-s}{4}\right)^2+4\right]^{1/2}
\end{equation*}
and it is attained for $k=1$. Note that (\ref{eq:aux-epsilon-eta-st})
is also satisfied when
\begin{equation}
\eta=\epsilon=\frac12\left(\frac{t-s}{2}+\left[\left(\frac{t-s}{2}\right)^2
  +4\right]^{1/2}\right)\,.
\end{equation}
Finally, one has
\begin{align*}
  J_c&\ge t-(\epsilon+\eta)=t-2\epsilon=\frac{t+s}{2}-\left[\left(\frac{t-s}{2}\right)^2
  +4\right]^{1/2}\\
&=\left\{\left(\frac{t+s}{2}\right)^2-\left[\left(\frac{t-s}{2}\right)^2
  +4\right]\right\}\left\{\frac{t+s}{2}+\left[\left(\frac{t-s}{2}\right)^2
  +4\right]^{1/2}\right\}^{-1}\\
&=(ts-4)\left\{\frac{t+s}{2}+\left[\left(\frac{t-s}{2}\right)^2
  +4\right]^{1/2}\right\}^{-1}
\end{align*}
\end{proof}

Now, we turn to the asymptotics of the decreasing generalized
eigenvectors of the semi-bounded block Jacobi matrix $J$. Denote by
$\mathcal{B}_\infty:= \lim_{n\to\infty}\mathcal{B}_n$. The following
arguments are heuristic. Let $\{\mu_q(n)\}_{q=1}^4$ be the
eigenvalues of $\mathcal{B}_n$, then, on the basis of a formal
Levinson type formula for $\lambda<0$, one has, for $n_0$ sufficiently large,
\begin{equation}
 \label{eq:asymptotic-heuristic}
  {u_n}^q\simeq\left(\prod_{k=n_0}^n\mu_q(k)\right)e_q\,,
\end{equation}
where $e_q$ are the eigenvectors of $\mathcal{B}_\infty$ (the proof of
similar asymptotic formulae are found in \cite{MR2579689}).

Using (\ref{eq:eigenvalues-transfer-matrix}) and
(\ref{eq:asymptotic-heuristic}), one obtains the following estimate of the
decreasing generalized eigenvectors
\begin{equation*}
  \norm{{u_n}}_{\complex^2}
\le {\rm const.}\prod_{k=n_0}^n\left[1-
\frac{(1-\epsilon)\sqrt{-\lambda(t+s)}}{2k^{\alpha/2}}\right]
\end{equation*}
for arbitrary small $\epsilon>0$ and $n\gg 1$. The last product can be
estimated from above by
\begin{equation*}
  C_{\epsilon_0}\exp\left[-(1-\epsilon_0)\frac{\sqrt{-\lambda(t+s)}}{2\left(1-\frac{\alpha}{2}\right)}n^{1-\alpha/2}\right]
\end{equation*}
for some constant $C_{\epsilon_0}$ and arbitrary
$\epsilon_0>\epsilon$. Since $st=4$, one has $\frac{\sqrt{s+t}}{2}\ge
1$. Thus, one can write
\begin{equation*}
  \norm{u_n}_{\complex^2}\le\exp\left[-(1-\epsilon_0)\frac{\sqrt{-\lambda}}{1-\frac{\alpha}{2}}n^{1-\alpha/2}\right]\,.
\end{equation*}
This estimate and the one obtained rigorously in
Theorem~\ref{thm:infinite-interval-in-spectrum} satisfy
\begin{align*}
  \exp(-\gamma(\lambda)\sum_{k=1}^{m-1}
\phi_\delta(\norm{A_k}_{B(\mathcal{K})}))&\asymp
\exp\left[-(1-\epsilon_0)\sqrt{-\lambda}
\sum_{k=1}^{n-1}\frac{1}{k^{\alpha/2}}\right]\\
&\asymp
\exp\left[-(1-\epsilon_0)\frac{\sqrt{-\lambda}}{1-\frac{\alpha}{2}}n^{1-\alpha/2}\right]\,.
\end{align*}
This formal reasoning shows sharpness of
Theorem~\ref{thm:infinite-interval-in-spectrum}, provided one chooses
$s$ and $t$ arbitrary close to $2$ and therefore
$\frac{1}{2}\sqrt{s+t}$ is arbitrary close to $1$.
\\[5mm]
\noindent\textbf{Section's concluding remarks}
  \begin{enumerate}[i]
  \item Weyl Theorem and the results of this section prove that
    $\sigma_{ess}(J)\subset\reals_+$.
  \item Lemma~\ref{lem:jc-bounded-below} and the decomposition
    $J=J_c+K$ show that $J$ is bounded from below and
    $\sigma(J)\cap\reals_-$ is discrete and can accumulate only at zero.
  \item One can prove that $\sigma_{ess}(J)=\reals_+$ by using the
    formal Levinson type asymptotics of solutions as an Ansatz for
    approximation of Weyl sequences corresponding to each
    $\lambda>0$. This idea is described in detail in \cite{MR2480099}.
  \item Concerning the assumption $st=4$, one can check that for
    $st>4$, $\sigma_{ess}(J)=\emptyset$ and if $st<4$, then
    $\sigma_{ess}(J)=\reals$. This explains the role of the condition
    $st=4$. In the case $s=t$, the matrix $J$ can be written as
    an orthogonal sum of
    two (unitarily equivalent) scalar Jacobi matrices by diagonalizing
    the matrix $\begin{pmatrix} 0&1\\ 1&0
    \end{pmatrix}$. In that case, the above results for $st>4$ and
    $st<4$ follow directly from \cite{MR1911189}. Moreover the result
    of (iii) for $s=t=2$ immediately follows from
    \cite{MR2550697}. Finally note that $J$, our class of block Jacobi
    matrices depending on parameters $s,t$, exhibits a spectral phase
    transition phenomenon of first kind (see \cite{MR1911189}) with
    the threshold corresponding to the condition $st=4$.
  \end{enumerate}

\section*{Acknowledgements}

JJ and SN have been supported by the National Science Centre - Poland, grant
no. 2013/09/B/ST1/04319. SN was also supported by grant RFBR 16-01-00443-a.
LOS has been supported by UNAM-DGAPA-PAPIIT IN110818 and
SEP-CONACYT CB-2015 254062. The authors thank the anonymous referee
whose pertinent comments led to an improved presentation of this work.

\def\cprime{$'$} \def\lfhook#1{\setbox0=\hbox{#1}{\ooalign{\hidewidth
  \lower1.5ex\hbox{'}\hidewidth\crcr\unhbox0}}}

\end{document}